\documentclass[11pt,a4paper]{article}
\usepackage[utf8]{inputenc}
\usepackage[T1]{fontenc}
\usepackage[a4paper]{geometry}
\usepackage{lmodern,textcomp,amsmath,amssymb,mathtools,braket,authblk,amsthm,csquotes}
\usepackage[english]{babel}

\newtheorem{theorem}{Theorem}[section]

\newtheorem{lemma}[theorem]{Lemma}
\newtheorem{proposition}[theorem]{Proposition}
\theoremstyle{definition}

\newtheorem{hypothesis}[theorem]{Hypothesis}
\theoremstyle{remark}
\newtheorem{remark}[theorem]{Remark}
\newtheorem{example}[theorem]{Example}
\AtBeginEnvironment{example}{%
	\pushQED{\qed}%
}
\AtEndEnvironment{example}{\popQED\endexample}
\AtBeginEnvironment{remark}{%
	\pushQED{\qed}%
}
\AtEndEnvironment{remark}{\popQED\endexample}

\newcommand{\hilb}{\mathcal{H}}
\newcommand{\fock}{\mathcal{F}}
\newcommand{\hfrak}{\mathfrak{H}}
\newcommand{\e}{\mathrm{e}}
\newcommand{\g}{\mathrm{g}}
\renewcommand{\Im}{\operatorname{Im}}

\renewcommand{\a}[1]{a\!\left(#1\right)}
\newcommand{\adag}[1]{a^\dag\!\left(#1\right)}
\newcommand{\ii}{\mathrm{i}}
\newcommand{\dOmega}{\mathrm{d}\Gamma(\omega)}

\makeatletter
\def\smalloverbrace#1{\mathop{\vbox{\m@th\ialign{##\crcr\noalign{\kern3\p@}%
				\tiny\downbracefill\crcr\noalign{\kern3\p@\nointerlineskip}%
				$\hfil\displaystyle{#1}\hfil$\crcr}}}\limits}
\makeatother

\usepackage[unicode=true,
bookmarks=true,bookmarksopen=false,
breaklinks=false,pdfborder={0 0 0},colorlinks=true]
{hyperref}

\usepackage{xcolor}
\definecolor{cblue}{rgb}{0.16, 0.32, 0.75}
\definecolor{cred}{rgb}{0.7, 0.11, 0.11}
\hypersetup{%
	,linkcolor=cred
	,citecolor=cblue
	,urlcolor=black
}

\usepackage[
backend=bibtex,
style=numeric-comp,
giveninits=true,
maxbibnames=299,
natbib=true,
doi=false,
isbn=false,
url=false,
sorting=none
]
{biblatex}
\renewbibmacro{in:}{}
\DeclareFieldFormat[misc]{title}{``#1''\isdot}
\DeclareFieldFormat{journaltitle}{#1\isdot}
\DeclareFieldFormat[article,periodical]{volume}{\mkbibbold{#1}}
\DeclareFieldFormat{pages}{#1}
\AtEveryBibitem{%
	\clearfield{number}}

\usepackage{orcidlink}

\title{\textbf{Self-adjointness of a class of multi-spin--boson models \\ with ultraviolet divergences}}
\author[$1,2,3,\star$]{Davide Lonigro\hspace{3pt}\orcidlink{0000-0002-0792-8122}\hspace{1pt}}

\affil[$1$]{\small Dipartimento di Matematica, Universit\`a degli Studi di Bari Aldo Moro, I-70125 Bari, Italy}
\affil[$2$]{\small Dipartimento di Fisica, Universit\`a degli Studi di Bari Aldo Moro, I-70126 Bari, Italy}
\affil[$3$]{\small Istituto Nazionale di Fisica Nucleare, Sezione di Bari, I-70126 Bari, Italy}
\affil[$\star$]{\small \texttt{davide.lonigro@ba.infn.it}}
\begin{document}
	
\vspace{-1cm}
\maketitle\vspace{-0.5cm}
	
\begin{abstract}
	We study a class of quantum Hamiltonian models describing a family of $N$ two-level systems (spins) coupled with a structured boson field of positive mass, with a rotating-wave coupling mediated by form factors possibly exhibiting ultraviolet divergences. Spin--spin interactions which do not modify the total number of excitations are also included. Generalizing previous results in the single-spin case, we provide explicit expressions for the self-adjointness domain and the resolvent of these models, both of them carrying an intricate dependence on the spin--field and spin--spin coupling via a family of concatenated propagators. This construction is also shown to be stable, in the norm resolvent sense, under approximations of the form factors via normalizable ones, for example an ultraviolet cutoff.
\end{abstract}
	
\section{Introduction}\label{sec:intro}

Consider an arbitrary number $N$ of distinguishable two-level systems (spins), whose free energy---possibly involving interactions between the spins---is described by a symmetric operator $K$ on a Hilbert space $\mathfrak{h}\simeq\mathbb{C}^{2^N}$; and a boson field, its free energy being represented by the second quantization $\dOmega$ of a self-adjoint operator $\omega$ on some single-particle Hilbert space $\hilb$, which is a self-adjoint operator on the Bose--Fock space $\fock(\hilb)$. For definiteness, we shall set $\hilb$ as the space of square-integrable functions, $\hilb=L^2_\mu(X)$, on some measure space $(X,\Sigma,\mu)$, and $\omega$ as the multiplication operator associated with some measurable function $X\ni k\mapsto\omega(k)\in\mathbb{R}$, playing the role of the dispersion relation of the field. We shall also require $\omega\geq m>0$ for some $m>0$. By setting
\begin{equation}\label{eq:model0}
	\hfrak=\mathfrak{h}\otimes\fock(\hilb)\simeq\mathbb{C}^{2^N}\otimes\fock(\hilb),\qquad H_{\rm free}=K\otimes\mathrm{I}_{\fock(\hilb)}+\mathrm{I}_{\mathfrak{h}}\otimes\dOmega,
\end{equation} 
where $\mathrm{I}_{\mathcal{K}}$ is the identity operator on the space $\mathcal{K}$, we shall consider operators on $\hfrak$ in the form
\begin{equation}\label{eq:model}
	H_{f_1,\dots,f_N}=H_{\rm free}+\sum_{j=1}^N\left(\sigma_j^-\otimes\adag{f_j}+\sigma_j^+\otimes\a{f_j}\right),
\end{equation}
where $\sigma_j^\pm$, $j=1,\dots,N$, are the operators implementing a transition between the excited and ground state of the $j$th spin, and $\a{f_j}$, $\adag{f_j}$ are the annihilation and creation operators associated with a family of functions $f_1,\dots,f_N\in\hilb$, called the \textit{form factors}; these are mutually adjoint closed operators $\a{f}$, $\adag{f}$ in $\fock(\hilb)$, with dense domain $\mathcal{D}(\a{f})=\mathcal{D}(\adag{f})$ (see e.g.~\cite{nelson1964interaction,reed1975fourier,bratteli1987operator}). An elementary application of the Kato--Rellich theorem shows that Eq.~\eqref{eq:model} defines a self-adjoint operator with domain $\mathcal{D}(H_{f_1,\dots,f_N})=\mathcal{D}(H_{\rm free})=\mathfrak{h}\otimes\mathcal{D}(\dOmega)$, see e.g.~\cite[Prop. 1.1]{arai1997existence}.

These operators belong to the class of (generalized) spin--boson models, whose relatively simple but rich structure allows for a detailed discussion of their mathematical properties, ranging from the spectrum to the existence and uniqueness of the ground state~\cite{hirokawa2001remarks,hubner1995spectral,hirokawa1999expression,arai1990asymptotic,amann1991ground,davies1981symmetry,fannes1988equilibrium,hubner1995radiative,reker2020existence,hasler2021existence,arai1997existence,arai2000essential,arai2000ground,falconi2015self,takaesu2010generalized,teranishi2015self,teranishi2018absence}, while also providing a simplified but realistic description of many natural phenomena of theoretical and practical interest. In particular, the model~\eqref{eq:model} can describe the interaction between two-level atoms with an electromagnetic field in the regime in which counter-rotating terms---those that do not preserve the total number of excitations---can be neglected, a procedure known as the rotating-wave approximation~\cite{agarwal1971rotating}, thus finding applications in the description of superradiance and subradiance phenomena~\cite{dicke1954coherence,gross1982superradiance,benedict2018super,van2013photon}, e.g. atom--photon bound states in the continuum~\cite{dorner2002laser,tufarelli2013dynamics,sanchez2017dynamical,gonzalez2017efficient,facchi2016bound,facchi2019bound,lonigro2021stationary,lonigro2021selfenergy}.

In this paper we shall address the \textit{ultraviolet problem} for the model~\eqref{eq:model0}---that is, its rigorous implementation as a self-adjoint operator on $\hfrak$ in the case in which $f_j\notin\hilb$ for some $j=1,\dots,N$, which (because of the assumption $\omega\geq m>0$) typically happens when $|f_j(k)|^2$ decreases ``too slowly'' as $|k|\to+\infty$. In this case, Eq.~\eqref{eq:model} does not define, \textit{per se}, a legitimate operator. We shall focus on the case in which the weaker normalization constraint
\begin{equation}\label{eq:minusone}
	\int_X\frac{|f_j(k)|^2}{\omega(k)}\,\mathrm{d}\mu<+\infty,\qquad j=1,\dots,N
\end{equation}
is imposed. In such a case, by interpreting Eq.~\eqref{eq:model} as an infinitesimal form perturbation of the free field operator via standard estimates on the annihilation operators, the existence of a self-adjoint model associated with Eq.~\eqref{eq:model}~\cite[Prop. 4.2]{lonigro2022generalized} follows as an application of the Kato--Lions--Lax--Milgram--Nelson (KLMN) theorem~\cite{teschl2009mathematical,simon2015quantum}; however, no information about the self-adjointness domain (in general, $\mathcal{D}(H_{f_1,\dots,f_N})\neq\mathcal{D}(H_{\rm free})$) nor the resolvent is retained. 

In the case $N=1$, an explicit construction of the self-adjoint operator associated with Eq.~\eqref{eq:model} was provided in~\cite[Theorem 5.5]{lonigro2022generalized} by interpreting the annihilation and creation operators as continuous operators on the scale of Hilbert spaces generated by $\dOmega$, through an explicit computation of the resolvent. The operator reduces to the standard (\textit{regular}) one when the form factor are normalizable; if not (\textit{singular} case), it can be approximated by sequences of regular ones---for example, by introducing an ultraviolet cutoff---in the norm resolvent sense. As expected, in the singular case one has $\mathcal{D}(H_{f_1,\dots,f_N})\neq\mathcal{D}(H_{\rm free})$. Furthermore, the result also holds for a form factor satisfying a constraint weaker than Eq.~\eqref{eq:minusone}, provided that a suitable renormalization of the energy of the excited state of the spin is performed.

Carrying forward this line of research, in this work we shall generalize this result to the case of an arbitrary number of spins (Theorem~\ref{thm}). Our construction, first hinted at in~\cite{lonigro2022renormalization}, is sufficiently general to accommodate spin--spin interactions, as long as they also preserve the total number of excitations of the system. This assumption will bring about a simple block-tridiagonal structure which will enable us to perform an explicit calculation of the resolvent. While the technique employed to obtain this result is essentially based on the same ideas at the root of the case $N=1$, its extension to the multi-spin case offers a number of additional mathematical intricacies that deserve a separate and thorough discussion.

Our results provide a contribution to the mathematical literature about the ultraviolet renormalization of similar models of matter--field interaction, to which extent a number of diverse techniques have been applied (see e.g.~\cite[Chapter 1]{lill2022time} for a recent overview). Typically, for sufficiently ``mild'' divergences (a typical example being the Fr\"ohlich polaron~\cite{frohlich1954electrons}), a quadratic form interpretation via the KLMN or the Friedrichs extension theorem is invoked; if not, cutoff renormalization methods (possibly with the aid of counterterms) are often applied, a paradigmatic example being the original treatment~\cite{nelson1964interaction} of divergences in the Nelson model. Some other techniques involve dressing transformations (possibly paired with cutoffs and/or counterterms), see e.g.~\cite{griesemer2016self,griesemer2018domain} where a partial characterization of the domain of the Fr\"ohlich polaron and the form domain of the Nelson model is obtained; or Feynman--Kac semigroups, which yield explicit expressions for the generated evolution~\cite{matte2017feynman,hinrichs2022feynman}; or, more recently, the use of extensions of the Fock space based on infinite tensor products (ITP) or extended state space (ESS) constructions~\cite{lill2020extended,lill2022implementing}. Most of these techniques have not been applied to spin--boson models, an exception being in~\cite{dam2020asymptotics}, where the existence of a renormalized spin--boson model \textit{not} having a rotating-wave interaction~\eqref{eq:model} and with an UV-divergent form factor not satisfying Eq.~\eqref{eq:minusone} is discussed---interestingly, in such a case the limiting Hamiltonian is necessarily trivial.

The approach applied in the present paper, while intrinsically model-dependent, is closer in spirit to two approaches based on the theory of self-adjoint extensions: the interior-boundary conditions (IBC) method~\cite{binz2021abstract} proposed in~\cite{teufel2016avoiding,teufel2021hamiltonians}, in which a Hamiltonian corresponding to a formal expression $H_{\rm free}+A+A^*$ is defined ab initio by means of abstract boundary conditions relating sectors with different number of particles; and the elegant abstract framework in~\cite{posilicano2020self} for self-adjoint operators associated with the formal expression above, formulated via singular perturbation theory~\cite{posilicano2008extensions}---that is, by first singularly perturbing $H_{\rm free}$, and then singularly perturbing the intermediate operator. The latter technique was applied to the Nelson model, yielding compatible results to the ones in~\cite{lampart2019nelson} (also cf.~\cite{lampart2018particle}) obtained via IBC. We will briefly revisit our results in the light of the latter ones at the end of the paper.

The paper is organized as follows. In Section~\ref{sec:regular} we will briefly describe the regular version of the model under investigation, and discuss its decomposition in a useful representation; in Section~\ref{sec:singular} we construct its generalization to the singular case.

\paragraph{Nomenclature.}\hspace{-.2cm}Throughout the paper, the scalar product on a Hilbert space $\mathcal{K}$ and its associated norm will be denoted via the following symbols:
\begin{equation}
	\Psi,\Phi\in\mathcal{K}\mapsto	\Braket{\Psi,\Phi}_{\mathcal{K}}\in\mathbb{C},\qquad\Phi\in\mathcal{K}\mapsto\left\|\Phi\right\|_{\mathcal{K}}\in\mathbb{R};
\end{equation}
the scalar product is chosen to be linear at the right, and antilinear at the left. The domain of an unbounded closed linear operator $T$ on $\mathcal{K}$ will be denoted by $\mathcal{D}(T)$. If $T$ admits an adjoint, it will be denoted by $T^*$. The complex conjugate of a number $z\in\mathbb{C}$ will be denoted by $\bar{z}$. Furthermore, given two Hilbert spaces $\mathcal{K},\mathcal{K}'$, the space of continuous (bounded) linear operators between $\mathcal{K}$ and $\mathcal{K}'$ will be denoted as $\mathcal{B}(\mathcal{K},\mathcal{K}')$; in the case $\mathcal{K}=\mathcal{K}'$, we will simply denote it as $\mathcal{B}(\mathcal{K})$.

\section{The model}\label{sec:regular}

We shall discuss the structure of the model~\eqref{eq:model} for \textit{regular} form factors $f_1,\dots,f_N\in\hilb$ by first describing its spin degrees of freedom (Section~\ref{subsec:spin}) and then introducing the coupling with the boson field (Section~\ref{subsec:enterboson}).

\subsection{Decomposition of the spin Hamiltonian}\label{subsec:spin}

The configuration space $\mathfrak{h}$ of $N$ two-level systems is the $2^N$-dimensional Hilbert space obtained, adopting a bra-ket notation, as
\begin{equation}
	\mathfrak{h}=\bigotimes_{j=1}^N\mathfrak{h}_j,\qquad\mathfrak{h}_j=\mathrm{Span}\{\ket{\e_j},\ket{\g_j}\}\simeq\mathbb{C}^2,
\end{equation}
with the normalized vectors $\ket{\e_j}$, $\ket{\g_j}$ respectively representing the excited and ground state of the $j$th spin. Let us introduce the following operator:
\begin{equation}\label{eq:nexc}
	N_{\rm exc}=\sum_{j=1}^N\left[\mathrm{I}_{\mathfrak{h}_1}\otimes\cdots\otimes\mathrm{I}_{\mathfrak{h}_{j-1}}\otimes\ket{\e_j}\!\!\bra{\e_j}\otimes\mathrm{I}_{\mathfrak{h}_{j+1}}\otimes\cdots\otimes\mathrm{I}_{\mathfrak{h}_{N}}\right]
\end{equation}
representing the number of spins in their excited state, whose spectrum comprises the eigenvalues $0,1,\dots,N$. Consequently, $\mathfrak{h}$ can be decomposed as the direct sum of its eigenspaces:\footnote{Notice that we shall use Greek indices $\nu,\nu',...$ to label the \textit{number} of excited spins, which ranges from $0$ to $N$, as opposed with the Latin indices $j,\ell,...$ which label the spins themselves and range from $1$ to $N$.}
\begin{equation}\label{eq:decomposition}
	\mathfrak{h}\simeq\bigoplus_{\nu=0}^N\mathfrak{h}^{(\nu)},\qquad \dim\mathfrak{h}^{(\nu)}=\binom{N}{\nu},
\end{equation}
with $N_{\rm exc}\mathfrak{h}^{(\nu)}=\nu\mathfrak{h}^{(\nu)}$. For example, $\mathfrak{h}^{(0)}$ is spanned by the only vector $\ket{g_1}\otimes\cdots\otimes\ket{g_N}$, $\mathfrak{h}^{(1)}$ is spanned by the $N$ vectors $\ket{g_1}\otimes\cdots\otimes\ket{\e_j}\otimes\cdots\ket{g_N}$, 
and so on. As such, the most general state $u\in\mathfrak{h}$ can be represented as
\begin{equation}\label{eq:representation}
	 u\simeq\left[\begin{array}{c}
		 u_0 \\\hline  u_1 \\\hline \vdots \\\hline  u_N
	\end{array}\right],\qquad  u_\nu\in\mathfrak{h}^{(\nu)},
\end{equation}
with $u_\nu$ representing a spin configuration in which $\nu$ spins are in their excited state, and $N-\nu$ are in their ground state.

Let $K\in\mathcal{B}(\mathfrak{h})$ be the operator representing the energy of the spins, possibly comprising spin--spin interactions. We shall make the following assumption on $K$:
\begin{hypothesis}\label{hyp1}
	$K\in\mathcal{B}(\mathfrak{h})$ is a nonnegative symmetric operator which preserves the number of excited spins: $[K,N_{\rm exc}]=0$, with $N_{\rm exc}$ as in Eq.~\eqref{eq:nexc}.
\end{hypothesis}
The nonnegativity of $K$, while not strictly required, will be assumed to ease the notation. $K$ is therefore block-diagonal with respect to the decomposition~\eqref{eq:decomposition} of $\mathfrak{h}$:
\begin{equation}\label{eq:blockdiagonal}
	K\simeq\left[\begin{array}{c|c|c|c}
		K_0 & & &\phantom{\ddots}\\\hline & K_1 & &\phantom{\ddots}\\\hline  & & K_2 &\phantom{\ddots}\\\hline &&&\ddots
	\end{array}\right],\qquad K_\nu\in\mathcal{B}\left(\mathfrak{h}^{(\nu)}\right).
\end{equation}

\begin{example}\label{ex:1}The simplest case of a spin Hamiltonian involving nontrivial $N_{\rm exc}$-preserving interactions is
\begin{equation}\label{eq:cappa}
K=\sum_{j=1}^N\left(e_j\,\sigma^+_j\sigma^-_j+g_j\,\sigma^-_j\sigma^+_j\right)+\sum_{j\neq\ell=1}^Nv_{j\ell}\,\sigma^+_j\sigma^-_\ell
\end{equation}
for some $e_j,g_j\in\mathbb{R}$, and $v_{j\ell}^*=v_{\ell j}\in\mathbb{C}$, and with the operators $\sigma^\pm_j$ defined by
\begin{eqnarray}\label{eq:sigmas1}
	\sigma^+_j&=&\mathrm{I}_{\mathfrak{h}_1}\otimes\cdots\otimes\mathrm{I}_{\mathfrak{h}_{j-1}}\otimes\ket{\e_j}\!\!\bra{\g_j}\otimes\mathrm{I}_{\mathfrak{h}_{j+1}}\otimes\cdots\otimes\mathrm{I}_{\mathfrak{h}_{N}},\\\label{eq:sigmas2}
	\sigma^-_j&=&\mathrm{I}_{\mathfrak{h}_1}\otimes\cdots\otimes\mathrm{I}_{\mathfrak{h}_{j-1}}\otimes\ket{\g_j}\!\!\bra{\e_j}\otimes\mathrm{I}_{\mathfrak{h}_{j+1}}\otimes\cdots\otimes\mathrm{I}_{\mathfrak{h}_{N}},
\end{eqnarray}
each representing a transition of the $j$th spin between its excited and ground state, and vice versa. The first term in Eq.~\eqref{eq:cappa} represents the free energy of the spins, with $e_j$ and $g_j$ respectively representing the excitation and ground energy of the $j$th spin; the second term encodes spin--spin interactions.
\end{example}

\subsection{Decomposition of the total Hamiltonian}\label{subsec:enterboson}

Let us now consider the Hilbert space $\hfrak=\mathfrak{h}\otimes\fock$, with $\hilb$ and $\fock(\hilb)\equiv\fock$ as defined in Section~\ref{sec:intro}. The decomposition~\eqref{eq:decomposition} of the spin subspace $\mathfrak{h}$ induces the following decomposition of $\hfrak$:
\begin{equation}\label{eq:decomp}
	\hfrak\simeq\bigoplus_{\nu=0}^N\hfrak^{(\nu)},\qquad\hfrak^{(\nu)}=\mathfrak{h}^{(\nu)}\otimes\fock,
\end{equation}
with $\hfrak^{(\nu)}$ being the space of all spin--field states such that $\nu$ out of $N$ spins are in their excited state. Any state in $\hfrak$ can thus be represented as
\begin{equation}\label{eq:representation2}
	\Psi\simeq\left[\begin{array}{c}
		\Psi_0 \\\hline \Psi_1 \\\hline \vdots \\\hline \Psi_N
	\end{array}\right],\qquad \Psi_\nu\in\hfrak^{(\nu)}.
\end{equation}
Given $f_1,f_2,\dots,f_N\in\hilb$, let us consider the (regular) multi-spin--boson model $H_{f_1,\dots,f_N}$ defined by Eq.~\eqref{eq:model}, which is a self-adjoint operator on $\hfrak$ with domain $\mathcal{D}(H_{f_1,\dots,f_N})=\mathcal{D}(H_{\rm free})=\mathfrak{h}\otimes\mathcal{D}(\dOmega)$. We shall require the following assumption about the single-particle operator $\omega$:
\begin{hypothesis}\label{hyp0}
	$\omega$ is a strictly positive self-adjoint operator on $\hilb$, with $\omega\geq m>0$.
\end{hypothesis}
We are thus focusing on the case of massive fields, in which case \textit{infrared} divergences are absent---only ultraviolet ones will be examined. The \textit{massless} case $m=0$, while of interest, would require a technically refined analysis which we will leave for future research: the construction of singular creation and annihilation operators developed in~\cite{lonigro2022generalized} that we shall employ in Section~\ref{sec:singular} relies on the fact that $\mathcal{D}(\dOmega)$ contains the domain of the number operator, which holds provided that $m>0$.

Therefore, $\dOmega$ is a nonnegative operator, and so is $H_{\rm free}$. We can write $H_{f_1,\dots,f_N}=H_{\rm free}+A+A^\dag$, with $H_{\rm free}$ as in Eq.~\eqref{eq:model0} and
\begin{equation}\label{eq:a}
	A=\sum_{j=1}^N\sigma_j^+\otimes\a{f_j},\qquad A^\dag=\sum_{j=1}^N\sigma_j^-\otimes\adag{f_j}.
\end{equation}
Recalling the decomposition~\eqref{eq:decomp}, the domain of $H_{f_1,\dots,f_N}$ can be similarly decomposed as
\begin{equation}
\mathcal{D}(H_{\rm free})\simeq\bigoplus_{\nu=0}^N\mathcal{D}^{(\nu)},\qquad \mathcal{D}^{(\nu)}=\mathfrak{h}^{(\nu)}\otimes\mathcal{D}(\dOmega)\subset\hfrak^{(\nu)},
\end{equation}
Since the spin Hamiltonian $K$ satisfies Hypothesis~\ref{hyp1}, for all $\nu=0,\dots,N$ we have $H_{\rm free}\mathcal{D}^{(\nu)}\subset\hfrak^{(\nu)}$, whereas, for all $\nu=0,\dots,N-1$, $A\mathcal{D}^{(\nu)}\subset\hfrak^{(\nu+1)}$ and $A^\dag\mathcal{D}^{(\nu+1)}\subset\hfrak^{(\nu)}$, thus bringing about a \textit{block-tridiagonal} decomposition for the Hamiltonian $H_{f_1,\dots,f_N}$:
\begin{equation}\label{eq:blocktrid}
	H_{f_1,\dots,f_N}\simeq\left[\begin{array}{c|c|c|c|c}
		H_0 & A_{0,1}^\dag& \phantom{\ddots}& \phantom{\ddots}&\\\hline 
		A_{0,1} & H_1 & A_{1,2}^\dag & \phantom{\ddots} &\\\hline  
		& A_{1,2} & H_2 & A_{2,3}^\dag &\phantom{\ddots} \\\hline
		& & A_{2,3} & H_3 & \ddots\\\hline 
		&&& \ddots& \ddots
	\end{array}\right],
\end{equation}
with the operators $H_\nu$, $A_{\nu,\nu+1}$ and $A_{\nu-1,\nu}^\dag$ respectively corresponding to the restrictions of $H_{\rm free}$, $A$ and $A^\dag$ to $\mathcal{D}^{(\nu)}$; in particular, $H_\nu=K_\nu\otimes\mathrm{I}_\fock+\mathrm{I}_{\mathfrak{h}^{(\nu)}}\otimes\dOmega$. We point out that the number of blocks scales linearly with $N$, while the dimension of $\mathfrak{h}$ scales exponentially with it.

\begin{example}[Case $N=2$]\label{ex:twoatom}
We shall briefly discuss how these decompositions come into play in the case of $N=2$ spins. The spin subspace $\mathfrak{h}$ decomposes as $\mathfrak{h}=\mathfrak{h}^{(0)}\oplus\mathfrak{h}^{(1)}\oplus\mathfrak{h}^{(2)}$, where
\begin{eqnarray}
\mathfrak{h}^{(0)}&=&\mathrm{Span}\,\{\ket{\g,\g}\};\\
\mathfrak{h}^{(1)}&=&\mathrm{Span}\,\{\ket{\e,\g},\ket{\g,\e}\};\\
\mathfrak{h}^{(2)}&=&\mathrm{Span}\,\{\ket{\e,\e}\},
\end{eqnarray}
where $\ket{\g,\g}:=\ket{\g_1}\otimes\ket{\g_2}$, $\ket{\e,\g}:=\ket{\e_1}\otimes\ket{\g_2}$, and so on. The Hamiltonian as in Eq.~\eqref{eq:cappa} reads
\begin{eqnarray}
	K&=&\sum_{j=1,2}\left(e_j\sigma_j^+\sigma_j^-+g_j\sigma_j^-\sigma_j^+\right)+v\sigma_1^+\sigma_2+v^*\sigma_2^+\sigma_1^-\nonumber\\
	&\simeq&\left[\begin{array}{c|cc|c}
		g_1+g_2&0 &0 &0\\\hline
		0&e_1+g_2&v&0\\
		0&v^*&g_1+e_2&0\\\hline
		0&0&0&e_1+e_2
	\end{array}\right],
\end{eqnarray}
with $e_j,g_j\in\mathbb{R}$ respectively corresponding to the energies of the excited and ground states of the $j$th spin, and $v\in\mathbb{C}$ modulating the spin--spin interaction; the distinction between sectors having the same number of excited spins has been stressed. 

Any vector in $\hfrak$ can be written as
\begin{equation}
	\Psi\simeq\left[\begin{array}{c}
		\Psi_{\rm gg} \\\hline \Psi_{\rm eg}\\ \Psi_{\rm ge} \\\hline \Psi_{\rm ee}
	\end{array}\right],\qquad\Psi_{\rm gg},\Psi_{\rm eg},\Psi_{\rm ge},\Psi_{\rm ee}\in\fock,
\end{equation}
with $\Psi_{xx'}$, $x,x'\in\{\rm e,g\}$, corresponding to the states of the boson field when the two spins are, respectively, in the $x$ and $x'$ state. Given two form factors $f_1,f_2\in\hilb$, the Hamiltonian $H_{f_1,f_2}$
acquires the following matrix structure:
\begin{equation}\label{eq:blocktrid_twoatom}
	H_{f_1,f_2}=\left[\begin{array}{c|cc|c}
		H_{\mathrm{gg}}&\adag{f_1}&\adag{f_2}& \\\hline
		\a{f_1}&H_{\mathrm{eg}}&v&\adag{f_2}\\
		\a{f_2}&v^*&H_{\mathrm{ge}}&\adag{f_1}\\\hline
		 &\a{f_2}&\a{f_1}&H_{\mathrm{ee}}
	\end{array}\right],
\end{equation}
where $H_{\rm xx'}=\dOmega+x_1+x'_2$, $x,x'\in\{e,g\}$. The structure~\eqref{eq:blocktrid} is thus recovered with
\begin{equation}
	H_0=H_{\rm gg},\qquad H_1=\left[\begin{array}{cc}
H_{\mathrm{eg}}&v\\
v^*&H_{\mathrm{ge}}
	\end{array}\right],\qquad H_2=H_{\rm ee}
\end{equation}
and
\begin{equation}\label{eq:a_twoatom}
A_{0,1}=\left[\begin{array}{c}
	\a{f_1}\\\a{f_2}
\end{array}\right],\qquad A_{1,2}=\left[\begin{array}{cc}
\a{f_2}&\a{f_1}
\end{array}\right].\qedhere
\end{equation}
\end{example}

\section{Resolvent and singular coupling}\label{sec:singular}

We shall now consider the case in which $f_1,\dots,f_N\notin\hilb$, but they satisfy the weaker condition~\eqref{eq:minusone}. In this case, the operators $A_{\nu,\nu+1},A_{\nu-1,\nu}^\dag$ as in Eq.~\eqref{eq:blocktrid} cease to be well-defined unbounded operators on $\hfrak^{(\nu)}$. We will follow the ideas developed in~\cite{lonigro2022generalized} to circumvent this problem.

\subsection{Some scales of Hilbert spaces}\label{subsec:scale}

Given a nonnegative self-adjoint operator $T$ on a Hilbert space $\mathcal{K}$, we will denote by \textit{$T$-scale} the nested family of Hilbert spaces $\{\mathcal{K}_s\}_{s\in\mathbb{R}}$, where $\mathcal{K}_s$ is the completion of $\bigcap_{n\in\mathbb{N}}\mathcal{D}(T^n)$ with respect to the norm
\begin{equation}
	\|\Phi\|_{\mathcal{K}_s}:=\left\|(T+1)^{s/2}\Phi\right\|_{\mathcal{K}}.
\end{equation}
Clearly, $\mathcal{K}_s\subset\mathcal{K}_{s'}$ whenever $s\geq s'$; in particular, $\mathcal{K}_0=\mathcal{K}$, and $\mathcal{K}_{+1}$ and $\mathcal{K}_{+2}$ coincide respectively with the form domain and operator domain of $T$. Notice that the ``$+1$'' in the definition above can be discarded whenever $T\geq\epsilon>0$. By construction, for every $r,s\in\mathbb{R}$, the operator $(T+1)^{r}$ can be continuously extended\footnote{
	With a slight abuse of notation, in such cases we will say that the operator can be ``interpreted'' either as an unbounded operator on $\mathcal{K}$, or as a bounded (continuous) operator between $\mathcal{K}_s$ and $\mathcal{K}_{s-2r}$.
} to an isometry between the spaces $\mathcal{K}_s$ and $\mathcal{K}_{s-2r}$. Consequently, for all $s>0$, the triple $\left(\mathcal{K}_{+s},\mathcal{K},\mathcal{K}_{-s}\right)$ is a Gel'fand triple~\cite{bohm1974rigged,de2005role,bohm1989dirac}, with the spaces $\mathcal{K}_{\pm s}$ being mutually dual with respect to the duality pairing
\begin{equation}\label{eq:pairing}
	(\Psi,\Phi)\in\mathcal{K}_{-s}\times\mathcal{K}_{+s}\mapsto\left(\Psi,\Phi\right)_{\mathcal{K}_{-s},\mathcal{K}_s}:=\Braket{(T+1)^{-s/2}\Psi,(T+1)^{+s/2}\Phi}_\mathcal{K},
\end{equation}
which is compatible with the scalar product on $\mathcal{K}$, that is, $(\Psi,\Phi)_{\mathcal{K}_{-s},\mathcal{K}_s}=\braket{\Psi,\Phi}_{\mathcal{K}}$ whenever $\Psi\in\mathcal{K}$; for this reason, with a slight abuse of notation we shall also use the latter symbol for the pairing, when the spaces $\mathcal{K}_{\pm s}$ are clear from the context. We refer to the literature, e.g.~\cite{albeverio2000singular,albeverio2007singularly,simon1995spectral}, for further details on the topic.

In our case, recalling that $\omega\geq m>0$ and $\dOmega\geq0$ by Hypothesis~\ref{hyp1}, we can define two scales of Hilbert spaces: the $\omega$-scale $\{\hilb_s\}_{s\in\mathbb{R}}$ associated with $\omega$, with
	\begin{equation}
		\|\psi\|^2_{\hilb_s}=\|\omega^{s/2}\psi\|^2_\hilb=\int\omega(k)^s|\psi(k)|^2\,\mathrm{d}\mu,
	\end{equation}
	and the $\dOmega$-scale $\{\fock_s\}_{s\in\mathbb{R}}$ associated with $\dOmega$, with
	\begin{equation}
		\|\Psi\|^2_{\fock_s}=\left\|\left(\dOmega+1\right)^{s/2}\Psi\right\|^2_\fock.
	\end{equation}
To ease the notation, we shall hereafter adopt the shorthand $\|\psi\|_{\hilb_s}\equiv\|\psi\|_s$ (and, in particular, $\|\psi\|_0\equiv\|\psi\|$) for the norms characterizing the $\omega$-scale. Clearly, the condition~\eqref{eq:minusone} is equivalent to $f_1,\dots,f_N\in\hilb_{-1}$.

For a measurable function $f\notin\hilb$, the annihilation operator $\a{f}$ still exists as an unbounded operator on $\fock$ but fails to have a densely defined adjoint. However, the following result holds:
\begin{proposition}[\cite{lonigro2022generalized}, Props. 3.4, 3.5 and 3.7]\label{prop:af}
	Let $\omega$ satisfy Hypothesis~\ref{hyp0}, and $f\in\hilb_{-s}$ for some $s\geq1$. Then the following statements are true:
	\begin{itemize}
		\item[(i)] the restriction of the annihilation operator to $\fock_s$ defines a continuous operator $\a{f}\in\mathcal{B}(\fock_{+s},\fock)$;
		\item[(ii)] its adjoint $\adag{f}:=\a{f}^*\in\mathcal{B}(\fock,\fock_{-s})$ with respect to the duality pairing between $\fock_s$ and $\fock_{-s}$ is a continuous operator whose action on $\Psi\in\fock_{+1}$ agrees with the creation operator;
		\item [(iii)] there exists a sequence $\{f^i\}_{i\in\mathbb{N}}\subset\hilb$ such that
		\begin{equation}
			\lim_{i\to\infty}\left\|a(f)-a(f^i)\right\|_{\mathcal{B}(\fock_s,\fock)}=0,\qquad 	\lim_{i\to\infty}\left\|a^\dag(f)-a^\dag(f^i)\right\|_{\mathcal{B}(\fock,\fock_{-s})}=0,
		\end{equation}
		and this happens if and only if $\|f-f^i\|_{-s}\to0$.
	\end{itemize}
\end{proposition}
The case of interest for us is $s=1$. The fact that, for $f\in\hilb$, the action of these ``singular'' operators agrees with that of the ``standard'' ones justifies the slight abuse of notation in using the same symbols $\a{f}$, $\adag{f}$. In these cases, $\a{f}$ and $\adag{f}$ can be interpreted either as unbounded operators on $\fock$ or as continuous operators respectively between $\fock_{+1}$ and $\fock$, and between $\fock$ and $\fock_{-1}$.

For our purposes, let us define a new scale of spaces $\{\hfrak_s\}_{s\in\mathbb{R}}$ by $\hfrak_{s}=\mathfrak{h}\otimes\fock_{s}$, with $\{\fock_{s}\}_{s\in\mathbb{R}}$ being the $\dOmega$-scale. These spaces can be decomposed analogously as in Eq.~\eqref{eq:decomposition}:
\begin{equation}
\hfrak_{s}=\bigoplus_{\nu=0}^N\hfrak^{(\nu)}_{s},	\qquad\hfrak^{(\nu)}_{s}=\mathfrak{h}^{(\nu)}\otimes\fock_{s}.
\end{equation}
Following a similar reasoning as in~\cite[Lemma 4.1]{lonigro2022generalized}, one readily shows that $\{\hfrak_s\}_{s\in\mathbb{R}}$, endowed with the topology inherited from $\{\fock_s\}_{s\in\mathbb{R}}$, coincides with the $H_{\rm free}$-scale of Hilbert spaces, since $H_{\rm free}$ differs from $\mathrm{I}_{\mathfrak{h}}\otimes\dOmega$ only by a bounded operator; analogously, for all $\nu=0,\dots,N$, $\{\hfrak_s^{(\nu)}\}_{s\in\mathbb{R}}$ coincides with the $H_\nu$-scale. Therefore, as a direct consequence of Prop.~\ref{prop:af}, we are provided with two families of continuous operators ($\nu=0,\dots,N-1$)
\begin{equation}\label{eq:annu}
A_{\nu,\nu+1}:\hfrak^{(\nu)}_{+1}\rightarrow\hfrak^{(\nu+1)},\qquad A_{\nu,\nu+1}^\dag:\hfrak^{(\nu+1)}\rightarrow\hfrak^{(\nu)}_{-1},
\end{equation}
which are also mutually adjoint with respect to the duality pairing between the two spaces.

With this interpretation, the block-tridiagonal operator in Eq.~\eqref{eq:blocktrid}, while no longer directly defining a self-adjoint operator on $\hfrak$, \textit{does} define a continuous operator between $\hfrak_{+1}$ and $\hfrak_{-1}$. Our goal will be to find a suitable domain $\mathcal{D}(H_{f_1,\dots,f_N})\subset\hfrak$ such that the restriction of said operator to $\mathcal{D}(H_{f_1,\dots,f_N})$ defines an unbounded operator on $\hfrak$ having all anticipated properties.

\subsection{A family of concatenated propagators}

The following lemma is a generalization of~\cite[Lemma 5.3]{lonigro2022generalized} to the multi-spin case.

\begin{lemma}\label{lemma:concatenated}
	Let $\omega$ be a single-particle operator satisfying Hypothesis~\ref{hyp0}, $K$ a spin Hamiltonian satisfying Hypothesis~\ref{hyp1}, $f_1,\dots,f_N\in\hilb_{-1}$, and for all $\nu=0,\dots,N-1$ let $H_\nu$, $A_{\nu,\nu+1}$ and $A_{\nu,\nu+1}^\dag$ as defined in Eq.~\eqref{eq:blocktrid}. 
	For all $z\in\mathbb{C}\setminus\mathbb{R}$, let $\mathcal{S}_\nu(z)$, $\mathcal{G}_\nu(z)$ be defined via the recurrence relations
	\begin{eqnarray}
		\mathcal{S}_0(z)=0,&&\\\label{eq:gnuz}
		\mathcal{G}_\nu(z)=H_\nu-z-\mathcal{S}_\nu(z),&\quad&\nu=0,\dots,N,\\
		\mathcal{S}_\nu(z)=A_{\nu-1,\nu}\,\mathcal{G}_{\nu-1}^{-1}(z)A_{\nu-1,\nu}^\dag,&\quad&\nu=1,\dots,N.\label{eq:snuz}
	\end{eqnarray}
	Then, for all $\nu$, 
	\begin{itemize}
		\item[(i)] $\mathcal{S}_\nu(z)$ is a bounded operator on $\hfrak^{(\nu)}$ which satisfies $\mathcal{S}_\nu(z)^*=\mathcal{S}_\nu(\bar{z})$ and, for all $\Psi\in\hfrak^{(\nu)}$,
			\begin{equation}\label{eq:herglotz}
				\frac{\Im\Braket{\Psi,\mathcal{S}_\nu(z)\Psi}_{\hfrak^{(\nu)}}}{\Im z}\geq0;
			\end{equation}
		\item[(ii)] $\mathcal{G}_\nu(z)$ is an unbounded operator on $\hfrak^{(\nu)}$, with domain $\mathcal{D}^{(\nu)}$, which satisfies $\mathcal{G}_\nu(z)^*=\mathcal{G}_\nu(\bar{z})$ and continuously extends to a bounded operator between $\hfrak^{(\nu)}_{+1}$ and $\hfrak^{(\nu)}_{-1}$ satisfying, for all\footnote{
			Here and in the remainder of the proposition, all quantities $\Braket{\Psi,\mathcal{G}_\nu(z)\Psi}_{\hfrak^{(\nu)}}$ must be interpreted, whenever $\Psi\in\hfrak^{(\nu)}_{+1}\setminus\mathcal{D}(H_{\nu})$, in the sense of the duality pairing between $\hfrak^{(\nu)}_{+1}$ and $\hfrak^{(\nu)}_{-1}$ with $\mathcal{G}_\nu(z)$ interpreted as a continuous operator between $\hfrak^{(\nu)}_{+1}$ and $\hfrak^{(\nu)}_{-1}$ (see the discussion at the start of Section~\ref{subsec:scale}).
		} 
		$\Psi\in\hfrak^{(\nu)}_{+1}$,
			\begin{equation}\label{eq:gnuim}
				\frac{\Im\Braket{\Psi,\mathcal{G}_\nu(z)\Psi}_{\hfrak^{(\nu)}}}{\Im z}\leq-\|\Psi\|^2_{\hfrak^{(\nu)}}\leq0;
			\end{equation}
			besides, it admits a bounded inverse $\mathcal{G}_\nu^{-1}(z)\in\mathcal{B}(\hfrak^{(\nu)})$, with operator norm
			\begin{equation}\label{eq:invgnorm}
				\left\|\mathcal{G}_\nu^{-1}(z)\right\|_{\mathcal{B}(\hfrak^{(\nu)})}\leq\frac{1}{|\!\Im z|},
			\end{equation}
			which continuously extends to a bounded operator between $\hfrak^{(\nu)}_{-1}$ and $\hfrak^{(\nu)}_{+1}$.
	\end{itemize}
\end{lemma}
\begin{proof}
	First of all, notice that the operators $H_\nu$ and $(H_\nu-z)^{-1}$ extend to continuous operators respectively from  $\hfrak^{(\nu)}_{+1}$ to $\hfrak^{(\nu)}_{-1}$, and vice versa. Indeed, by the definition of the norms on the $H_\nu$-scale (see Section~\ref{subsec:scale}), this happens if and only if the operators
	\begin{equation}
		(H_\nu+1)^{-1/2}H_\nu(H_\nu+1)^{-1/2},\qquad 	(H_\nu+1)^{1/2}(H_\nu-z)^{-1}(H_\nu+1)^{1/2}
	\end{equation}
can be extended to bounded operators on $\hfrak^{(\nu)}$, which is clearly true. Since $\mathcal{G}_0(z)=H_0-z$, the claim holds for $\nu=0$. We shall prove it for $\nu\geq1$ by induction.
	
	Suppose that the claim is true for $\nu-1$. Then $\mathcal{S}_\nu(z)$ is a bounded operator on $\hfrak^{(\nu)}$ since, taking into account Eq.~\eqref{eq:annu}, we have
	\begin{equation}
		\mathcal{S}_\nu(z):\quad\hfrak^{(\nu)}
		\overset{A_{\nu,\nu+1}^\dag}{\xrightarrow{\makebox[1cm]{}}}
		\hfrak^{(\nu-1)}_{-1}
		\overset{\mathcal{G}_{\nu-1}^{-1}(z)}{\xrightarrow{\makebox[1.5cm]{}}}
		\hfrak^{(\nu-1)}_{+1}
		\overset{A_{\nu,\nu+1}}{\xrightarrow{\makebox[1cm]{}}}
		\hfrak^{(\nu)};
	\end{equation}
	the property $\mathcal{S}_\nu(z)^*=\mathcal{S}_\nu(\bar{z})$ also follows directly from the fact that $\mathcal{G}_{\nu-1}(z)^*=\mathcal{G}_{\nu-1}(\bar{z})$, and from the fact that $A_{\nu-1,\nu}$ and $A_{\nu-1,\nu}^\dag$ are mutually adjoint. Besides, given $\Psi\in\hfrak^{(\nu)}$, we have
	\begin{eqnarray}\label{eq:similar}
		\Im\Braket{\Psi,\mathcal{S}_\nu(z)\Psi}_{\hfrak^{(\nu)}}&=&\Im\Braket{\Psi,A_{\nu-1,\nu}\,\mathcal{G}_{\nu-1}^{-1}(z)A_{\nu-1,\nu}^\dag\Psi}_{\hfrak^{(\nu)}}\nonumber\\
		&=&\frac{1}{2\ii}\Braket{
			\Psi,A_{\nu-1,\nu}\Big\{\mathcal{G}_{\nu-1}^{-1}(z)-\mathcal{G}_{\nu-1}^{-1}(\bar{z})\Big\}A_{\nu-1,\nu}^\dag\Psi
		}_{\hfrak^{(\nu)}}\nonumber\\
		&=&\frac{1}{2\ii}\Braket{
		\Psi,A_{\nu-1,\nu}\,\mathcal{G}_{\nu-1}^{-1}(z)\Big\{\mathcal{G}_{\nu-1}(\bar{z})-\mathcal{G}_{\nu-1}(z)\Big\}\mathcal{G}_{\nu-1}^{-1}(\bar{z})A_{\nu-1,\nu}^\dag\Psi
		}_{\hfrak^{(\nu)}}\nonumber\\
		&=&-\Im\Braket{\mathcal{G}_{\nu-1}^{-1}(\bar{z})A_{\nu-1,\nu}^\dag\Psi,\mathcal{G}_{\nu-1}(z)\,\mathcal{G}_{\nu-1}^{-1}(\bar{z})A_{\nu-1,\nu}^\dag\Psi}_{\hfrak^{(\nu)}}.
	\end{eqnarray}
	Noticing that $\mathcal{G}_{\nu-1}^{-1}(\bar{z})A_{\nu-1,\nu}^\dag\Psi\in\hfrak^{(\nu-1)}_{+1}$, by Eq.~\eqref{eq:gnuim} we have $\Im\Braket{\Psi,\mathcal{S}_\nu(z)\Psi}_{\hfrak^{(\nu)}}\geq0$ whenever $\Im z>0$ and vice versa, completing the proof of (i).
	
	We must now prove that $\mathcal{G}_\nu(z)$, as defined in Eq.~\eqref{eq:gnuz}, has the desired properties. Clearly, being a bounded perturbation of $H_\nu-z$, the operator is well-defined on the dense domain $\mathcal{D}^{(\nu)}$ and its adjoint is $\mathcal{G}_\nu(z)^*=\mathcal{G}_\nu(\bar{z})$; for the same reason, since $H_\nu$ can be extended to a continuous operator from $\hfrak^{(\nu)}_{+1}$ to $\hfrak^{(\nu)}_{-1}$, the same holds for $\mathcal{G}_\nu(z)$. The second property follows directly from the definition~\eqref{eq:gnuz} of $\mathcal{G}_\nu(z)$:
	\begin{equation}
		\Im\Braket{\Psi,\mathcal{G}_\nu(z)\Psi}_{\hfrak^{(\nu)}}=-\Im z\,\|\Psi\|^2_{\hfrak^{(\nu)}}-\Im\Braket{\Psi,\mathcal{S}_\nu(z)\Psi}_{\hfrak^{(\nu)}},
	\end{equation}
	which, by (i), implies
	\begin{equation}\label{eq:using}
		\frac{\Im\Braket{\Psi,\mathcal{G}_\nu(z)\Psi}_{\hfrak^{(\nu)}}}{\Im z}=-\|\Psi\|^2_{\hfrak^{(\nu)}}-\frac{\Im\Braket{\Psi,\mathcal{S}_\nu(z)\Psi}_{\hfrak^{(\nu)}}}{\Im z}\leq-\|\Psi\|^2_{\hfrak^{(\nu)}}.
	\end{equation}
	To prove that $\mathcal{G}_\nu(z)$ has a bounded inverse, notice that
	\begin{equation}
		\left|\Braket{\Psi,\mathcal{G}_\nu(z)\Psi}_{\hfrak^{(\nu)}}\right|\geq\left|\Im\Braket{\Psi,\mathcal{G}_\nu(z)\Psi}_{\hfrak^{(\nu)}}\right|\geq|\!\Im z|\|\Psi\|^2_{\hfrak^{(\nu)}},
	\end{equation}
	whence, applying the Cauchy--Schwarz inequality and recalling that $\mathcal{G}_\nu(z)^*=\mathcal{G}_\nu(\bar{z})$,
	\begin{equation}
		\left\|\mathcal{G}_\nu(z)\Psi\right\|_{\hfrak^{(\nu)}}\geq|\!\Im z|\|\Psi\|_{\hfrak^{(\nu)}},\qquad \left\|\mathcal{G}_\nu(z)^*\Psi\right\|_{\hfrak^{(\nu)}}\geq|\!\Im z|\|\Psi\|_{\hfrak^{(\nu)}};
	\end{equation}
	these conditions imply that $\mathcal{G}_\nu(z)$ has a bounded inverse in $\hfrak^{(\nu)}$ whose operator norm satisfies Eq.~\eqref{eq:invgnorm}, see e.g.~\cite[Theorem 3.3.2]{derezinski2013unbounded}.
	
	Finally, again by the definition of the norms of the $H_\nu$-scale, $\mathcal{G}_\nu^{-1}(z)$ extends to a continuous operator from $\hfrak^{(\nu)}_{-1}$ to $\hfrak^{(\nu)}_{+1}$ if and only if the operator on $\hfrak^{(\nu)}$, with domain $\hfrak^{(\nu)}_{+1}$, acting as
	\begin{equation}
		(H_\nu+1)^{1/2}\mathcal{G}_\nu^{-1}(z)(H_\nu+1)^{1/2}
	\end{equation}
	extends to a bounded operator in $\hfrak^{(\nu)}$. Now, using Eq.~\eqref{eq:gnuz} of $\mathcal{G}_\nu(z)$ and applying the second resolvent identity,
	\begin{eqnarray}
		(H_\nu+1)^{1/2}\mathcal{G}_\nu^{-1}(z)(H_\nu+1)^{1/2}&=&	(H_\nu+1)^{1/2}(H_\nu-z)^{-1}(H_\nu+1)^{1/2}\nonumber\\&&-(H_\nu+1)^{1/2}\mathcal{G}_\nu^{-1}(z)\,\mathcal{S}_\nu(z)(H_\nu-z)^{-1}(H_\nu+1)^{1/2}.
	\end{eqnarray}
	The first term clearly extends to a bounded operator on $\hfrak^{(\nu)}$; likewise, the operator $(H_\nu-z)^{-1}(H_\nu+1)^{1/2}$ extends to a bounded operator on $\hfrak^{(\nu)}$, whence the claim follows from the boundedness of $\mathcal{S}_\nu(z)$ and $\mathcal{G}_\nu^{-1}(z)$ and the fact that the latter has values in $\mathcal{D}^{(\nu)}\subset\hfrak^{(\nu)}$. 
\end{proof}
Notice that both functions $z\in\mathbb{C}\setminus\mathbb{R}\mapsto\mathcal{S}_\nu(z),\mathcal{G}_\nu^{-1}(z)\in\mathcal{B}(\hilb^{(\nu)})$ are operator-valued Nevanlinna (or Herglotz) functions~\cite{gesztesy2000matrix,gesztesy2001some}.

\begin{example}[Case $N=2$]
	We shall take a look at the explicit structure of the propagators in the case $N=2$, first discussed in~Example~\ref{ex:twoatom}, in which $\nu=0,1,2$. For simplicity, let us consider the case in which spin--spin interactions are absent---that is, $v=0$ in Eq.~\eqref{eq:blocktrid_twoatom}. For $\nu=0$, we simply have $\mathcal{S}_0(z)=0$ and $\mathcal{G}_0(z)=H_{\rm gg}-z$; for $\nu=1$,
	\begin{eqnarray}\label{eq:n2s1}
		\mathcal{S}_1(z)&=&\left[
		\begin{array}{c}
			\a{f_1}\\\a{f_2}
			\end{array}
			\right]\mathcal{G}_0^{-1}(z)\left[\begin{array}{cc}
		\adag{f_1}&\adag{f_2}
			\end{array}\right]\nonumber\\
		&=&\renewcommand{\arraystretch}{1.5}\left[\begin{array}{cc}
		\a{f_1}\frac{1}{H_{\rm gg}-z}\adag{f_1}&\a{f_1}\frac{1}{H_{\rm gg}-z}\adag{f_2}\\
		\a{f_2}\frac{1}{H_{\rm gg}-z}\adag{f_1}&\a{f_2}\frac{1}{H_{\rm gg}-z}\adag{f_2}
		\end{array}\right],
	\end{eqnarray}
	\begin{eqnarray}\label{eq:n2g1}
		\mathcal{G}_1(z)&=&H_1-z-\mathcal{S}_1(z)\nonumber\\&=&\renewcommand{\arraystretch}{1.5}\left[\begin{array}{cc}
		H_{\rm ge}-z-\a{f_1}\frac{1}{H_{\rm ee}-z}\adag{f_1}&-\a{f_1}\frac{1}{H_{\rm ee}-z}\adag{f_2}\\
		-\a{f_2}\frac{1}{H_{\rm ee}-z}\adag{f_1}&H_{\rm eg}-z-\a{f_2}\frac{1}{H_{\rm ee}-z}\adag{f_2}
		\end{array}\right].
	\end{eqnarray}
	Note that the diagonal terms of both $\mathcal{S}_1(z)$ and $\mathcal{G}_1(z)$ respectively correspond to the values of the self-energy and the propagator in the case $N=1$~\cite{lonigro2022generalized}.	Finally, for $\nu=2$,
	\begin{eqnarray}
		\mathcal{S}_2(z)&=&\left[
		\begin{array}{cc}\a{f_2}&\a{f_1}\end{array}
		\right]\mathcal{G}_1^{-1}(z)\left[\begin{array}{c}
			\adag{f_2}\\\adag{f_1}
		\end{array}\right],\\
	\mathcal{G}_2(z)&=&H_{\rm ee}-z-\left[
	\begin{array}{cc}\a{f_2}&\a{f_1}\end{array}
		\right]\mathcal{G}_1^{-1}(z)\left[\begin{array}{c}
		\adag{f_2}\\\adag{f_1}
		\end{array}\right].
	\end{eqnarray}
	Clearly, the dependence of both operators $\mathcal{S}_2(z)$, $\mathcal{G}_2(z)$ on the inverse of $\mathcal{G}_1(z)$ makes their exact expression already involved for $N=2$---even in the absence of spin--spin interactions. Still, these expressions may be used as the starting point for perturbative expansions potentially useful for spectral estimates.
	\end{example}

\subsection{The singular multi-spin--boson model}\label{subsec:results}

In order to present the main result of the work, we need to introduce yet some other families of $z$-dependent operators.

\begin{lemma}\label{lemma:gammas}
	Given $z\in\mathbb{C}\setminus\mathbb{R}$, the operators ($\nu'>\nu$)
\begin{eqnarray}\label{eq:gammanu3}
	\gamma_{\nu,\nu'}(z):\hfrak^{(\nu)}_{-1}\rightarrow\hfrak^{(\nu')},&\qquad&\gamma_{\nu,\nu'}(z)=\gamma_{\nu'-1,\nu'}(z)\cdots\gamma_{\nu+1,\nu+2}(z)\,\gamma_{\nu,\nu+1}(z),\\\label{eq:gammanu4}
	\gamma^\dag_{\nu,\nu'}(z):\hfrak^{(\nu')}\rightarrow\hfrak^{(\nu)}_{+1},&\qquad&\gamma^\dag_{\nu,\nu'}(z)=\gamma^\dag_{\nu,\nu+1}(z)\,\gamma^\dag_{\nu+1,\nu+2}(z)\cdots\gamma^\dag_{\nu'-1,\nu'}(z),
\end{eqnarray}
where
\begin{eqnarray}\label{eq:gammanu}
	\gamma_{\nu,\nu+1}(z):\hfrak^{(\nu)}_{-1}\rightarrow\hfrak^{(\nu+1)},&\qquad&\gamma_{\nu,\nu+1}(z):=-A_{\nu,\nu+1}\mathcal{G}^{-1}_{\nu}(z),\\\label{eq:gammanu2}
	\gamma^\dag_{\nu,\nu+1}(z):\hfrak^{(\nu+1)}\rightarrow\hfrak^{(\nu)}_{+1},&\qquad&\gamma^\dag_{\nu,\nu+1}(z):=-\mathcal{G}^{-1}_{\nu}(z)A^\dag_{\nu,\nu+1},
\end{eqnarray}
are well-defined; besides, if $f_1,\dots,f_N\in\hilb$, then, for all $\Phi\in\mathcal{D}^{(\nu')}$, $\gamma^\dag_{\nu,\nu'}(z)\Phi\in\mathcal{D}^{(\nu)}$.
\end{lemma}

\begin{proof}
The first statement is a simple consequence of the fact that, by Lemma~\ref{lemma:concatenated}, $\mathcal{G}^{-1}_\nu(z):\hfrak^{(\nu)}_{-1}\rightarrow\hfrak^{(\nu)}_{+1}$. Indeed,
\begin{eqnarray}
	\gamma_{\nu,\nu+1}(z):&\quad&
	\hfrak^{(\nu)}_{-1}
	\overset{\mathcal{G}^{-1}_\nu(z)}{\xrightarrow{\makebox[1cm]{}}}
	\hfrak^{(\nu)}_{+1}
	\overset{A_{\nu,\nu+1}}{\xrightarrow{\makebox[1cm]{}}}
	\hfrak^{(\nu+1)};\\
		\gamma^\dag_{\nu,\nu+1}(z):&\quad&
	\hfrak^{(\nu+1)}
	\overset{A^\dag_{\nu,\nu+1}}{\xrightarrow{\makebox[1cm]{}}}
	\hfrak^{(\nu)}_{-1}
	\overset{\mathcal{G}^{-1}_\nu(z)}{\xrightarrow{\makebox[1cm]{}}}
	\hfrak^{(\nu)}_{+1}.
\end{eqnarray}
Iterating this argument, the first claim is readily proven. The second one follows from the fact that, if $\Phi\in\mathcal{D}^{(\nu+1)}\subset\hfrak^{(\nu+1)}_{+1}$ and all form factors are regular, then the creation operators in the definition of $A^\dag_{\nu,\nu+1}$ can be interpreted in the standard sense, whence $A^\dag_{\nu,\nu+1}\Phi\in\hfrak^{(\nu)}$; thus $\mathcal{G}^{-1}_\nu(z)$ can, as well, interpreted as a bounded operator on $\hfrak^{(\nu)}$ with range equal to the domain of $\mathcal{G}_\nu(z)$, which is $\mathcal{D}^{(\nu)}$; hence $\gamma_{\nu,\nu+1}^\dag(z)$ maps $\mathcal{D}^{(\nu+1)}$ in $\mathcal{D}^{(\nu)}$. Applying an analogous reasoning to $\gamma_{\nu,\nu'}^\dag(z)$ completes the proof.
\end{proof}

We can thus define
\begin{eqnarray}
	\mathcal{L}^{-1}(z)&=&\left[\begin{array}{c|c|c|c|c}
		\mathrm{I} & \phantom{\ddots} & \phantom{\ddots} & \phantom{\ddots}&\\\hline 
		\gamma_{0,1}(z) & \mathrm{I}  & \phantom{\ddots} & \phantom{\ddots} &\\\hline  
		\gamma_{0,2}(z) & \gamma_{1,2}(z) & \mathrm{I} & \phantom{\gamma_{1,2}^\dag(z)} &\phantom{\ddots} \\\hline
		\gamma_{0,3}(z)	& \gamma_{1,3}(z) & \gamma_{2,3}(z) & \mathrm{I} & \phantom{\ddots\ddots} \\\hline 
		\vdots & \vdots & \vdots& \vdots& \ddots
	\end{array}\right],\label{eq:l(z)inv}\\\nonumber\\
	\mathcal{U}^{-1}(z)&=&\left[\begin{array}{c|c|c|c|c}
		\mathrm{I} & \gamma_{0,1}^\dag(z)& \gamma_{0,2}^\dag(z) & \gamma_{0,3}^\dag(z) & \cdots\\\hline 
		\phantom{\gamma_{0,1}(z)} & \mathrm{I} & \gamma_{1,2}^\dag(z) & \gamma_{1,3}^\dag(z) & \cdots\\\hline  
		& \phantom{\ddots} & \mathrm{I} & \gamma_{2,3}^\dag(z) &\cdots \\\hline
		& & \phantom{\ddots} & \mathrm{I} & \cdots\\\hline 
		&&& & \phantom{..}\ddots\phantom{..}
	\end{array}\right],\label{eq:u(z)inv}
\end{eqnarray}
with $\mathcal{L}^{-1}(z):\hfrak_{-1}\rightarrow\hfrak_{-1}$ and $\mathcal{U}^{-1}(z):\hfrak_{+1}\rightarrow\hfrak_{+1}$; notice that, to ease the notation, all subscripts on the identity operators in Eqs.~\eqref{eq:l(z)inv}--\eqref{eq:u(z)inv} are left understood.

We can finally provide the desired generalization of~\cite[Theorem 5.5]{lonigro2022generalized} to the multi-spin case:
\begin{theorem}\label{thm}
		Let $\omega$ be a single-particle operator satisfying Hypothesis~\ref{hyp0}, $K$ a spin Hamiltonian satisfying Hypothesis~\ref{hyp1}, $f_1,\dots,f_N\in\hilb_{-1}$
		, and for all $\nu=0,\dots,N$ let $H_\nu$, $A_{\nu,\nu+1}$ and $A_{\nu,\nu+1}^\dag$ as defined in Eq.~\eqref{eq:blocktrid}.
		Let $H_{f_1,\dots,f_N}$ be the operator on $\hfrak$ acting as in Eq.~\eqref{eq:blocktrid} on the domain		
		\begin{equation}\label{eq:domain}
			\mathcal{D}\!\left(H_{f_1,\dots,f_N}\right)=\left\{\Phi=\mathcal{U}^{-1}(z_0)\tilde{\Phi}:\;\tilde{\Phi}\in\mathcal{D}(H_{\rm free}),\;z_0\in\mathbb{C}\setminus\mathbb{R}\right\}\subset\hfrak_{+1},
		\end{equation}
	for some fixed $z_0\in\mathbb{C}\setminus\mathbb{R}$, with $\mathcal{U}^{-1}(z)$ as in Eq.~\eqref{eq:u(z)inv}. Then the following facts hold:
		\begin{itemize}
			\item [(i)] for all $f_1,\dots,f_N\in\hilb_{-1}$, $H_{f_1,\dots,f_N}$ is a self-adjoint operator on $\hfrak$, with resolvent
			\begin{equation}
				\left(H_{f_1,\dots,f_N}-z\right)^{-1}=\mathcal{U}^{-1}(z)\mathcal{G}^{-1}(z)\mathcal{L}^{-1}(z),
			\end{equation}
		with $\mathcal{L}^{-1}(z)$, $\mathcal{U}^{-1}(z)$ as in Eqs.~\eqref{eq:l(z)inv}--\eqref{eq:u(z)inv}, and $\mathcal{G}^{-1}(z)=\mathrm{diag}\left\{\mathcal{G}^{-1}_0(z),\dots,\mathcal{G}^{-1}_N(z)\right\}$;
			\item [(ii)] if $f_1,\dots,f_N\in\hilb$, $H_{f_1,\dots,f_N}$ coincides with the (regular) multi-spin--boson model with domain $\mathcal{D}(H_{\rm free})$;
			\item [(iii)] given $f_1,\dots,f_N\in\hilb_{-1}\setminus\hilb$, there exist sequences $\{f^i_1\}_{i\in\mathbb{N}},\ldots,\{f^i_N\}_{i\in\mathbb{N}}\subset\hilb$ such that
			\begin{equation}
				H_{f^i_1,\dots,f^i_N}\to H_{f_1,\dots,f_N}\;\;(i\to\infty)\qquad\text{in the norm resolvent sense,}
			\end{equation}
		and this happens if and only if $\|f^i_j-f_j\|_{-1}\to0$ for all $j=1,\dots,N$.
		\end{itemize}
\end{theorem}

\begin{remark}\label{rem:independent}
	The domain~\eqref{eq:domain} (and thus the corresponding operator) is in fact independent of the choice of $z_0$. Indeed, take $\Phi=\mathcal{U}^{-1}(z_0)\tilde{\Phi}$; then, for any other $z\in\mathbb{C}\setminus\mathbb{R}$,
\begin{equation}
	\Phi=\mathcal{U}^{-1}(z)\mathcal{U}(z)\mathcal{U}^{-1}(z_0)\tilde{\Phi};
\end{equation}
but since, by Lemma~\ref{lemma:gammas}, all operators $\gamma^\dag_{\nu,\nu'}(z)$ map $\mathcal{D}^{(\nu')}$ into $\mathcal{D}^{(\nu)}$, then both $\mathcal{U}(z)$ and $\mathcal{U}^{-1}(z_0)$ map $\mathcal{D}(H_{\rm free})$ into itself; whence $\Phi=\mathcal{U}^{-1}(z)\tilde{\Phi}'$ with $\tilde{\Phi}'\in\mathcal{D}(H_{\rm free})$.
\end{remark}

\begin{proof}[Proof of Theorem~\ref{thm}]
Given $f_1,\dots,f_N\in\hilb_{-1}$, the operator $H_{f_1,\dots,f_N}$ in Eq.~\eqref{eq:blocktrid} can be interpreted as a continuous operator between the spaces $\hfrak_{+1}$ and $\hfrak_{-1}$. We shall solve the equation
\begin{equation}\label{eq:extended}
	\left(H_{f_1,\dots,f_N}-z\right)\Phi=\Psi,\qquad \Phi\in\hfrak_{+1},\;\Psi\in\hfrak_{-1},
\end{equation}
for $z\in\mathbb{C}\setminus\mathbb{R}$. This requires computing the inverse of an operator-valued block-tridiagonal matrix, which will be done by exploiting a convenient version of the LU factorization of block-tridiagonal matrices (see e.g.~\cite{meurant1992review}). We have
	\begin{equation}\label{eq:lu}
			H_{f_1,\dots,f_N}-z=\mathcal{L}(z)\mathcal{G}(z)\mathcal{U}(z),
		\end{equation}
where $\mathcal{G}(z)=\mathrm{diag}\left\{\mathcal{G}_0(z),\mathcal{G}_1(z),\ldots\right\}$ and
\begin{eqnarray}
	\mathcal{L}(z)&=&\left[\begin{array}{c|c|c|c|c}
		\mathrm{I} 		 & \phantom{-\gamma_{0,1}(z)} & \phantom{-\gamma_{0,1}(z)} & \phantom{-\gamma_{0,1}(z)} & \\\hline 
		-\gamma_{0,1}(z) & \mathrm{I}                 & \phantom{\ddots}           & \phantom{\ddots}           &\\\hline  
		                 & -\gamma_{1,2}(z)           & \mathrm{I}                 & \phantom{-\ddots\ddots}    &\phantom{\ddots} \\\hline
		                 &                            & -\gamma_{2,3}(z)           & \mathrm{I}                 & \phantom{\ddots\ddots} \\\hline 
		                 &                            &                            & \ddots                     & \ddots
	\end{array}\right],\\\nonumber\\
	\mathcal{U}(z)&=&\left[\begin{array}{c|c|c|c|c}
		\mathrm{I}                 & -\gamma_{0,1}^\dag(z)   & \phantom{-\gamma_{0,1}(z)} & \phantom{-\gamma_{0,1}(z)} & \phantom{\ddots\ddots}\\\hline 
		\phantom{-\gamma_{0,1}(z)} & \mathrm{I}              & -\gamma_{1,2}^\dag(z)      & \phantom{\ddots\ddots}     &\\\hline  
		                           & \phantom{\ddots}        & \mathrm{I}                 & -\gamma_{2,3}^\dag(z)      &\phantom{\ddots} \\\hline
		                           &                         & \phantom{\ddots}           & \mathrm{I}                 & \ddots\\\hline 
		                           &                         &                            &                            & \ddots
	\end{array}\right],\label{eq:uz}
\end{eqnarray}
with $\mathcal{U}(z):\hfrak_{+1}\rightarrow\hfrak_{+1}$ and $\mathcal{L}(z):\hfrak_{-1}\rightarrow\hfrak_{-1}$ being the operator-valued block-triangular matrices appearing in the LU decomposition, and the operators $\gamma_{\nu,\nu+1}(z)$, $\gamma^\dag_{\nu,\nu+1}(z)$ as given by Eqs.~\eqref{eq:gammanu}--\eqref{eq:gammanu2}. Notice that, in particular,
\begin{equation}\label{eq:notice}
\mathcal{L}(z)\hfrak\subset\hfrak,\qquad \mathcal{L}^{-1}(z)\hfrak\subset\hfrak.
\end{equation}
Because of their triangular structure, the inverses of $\mathcal{L}(z)$ and $\mathcal{U}(z)$ exist and are given by Eqs.~\eqref{eq:l(z)inv}--\eqref{eq:u(z)inv}, whence the operator
\begin{equation}\label{eq:resolvent}
	\mathcal{U}^{-1}(z)\mathcal{G}^{-1}(z)\mathcal{L}^{-1}(z):\hfrak_{-1}\rightarrow\hfrak_{+1},
\end{equation}
with $\mathcal{L}^{-1}(z)$ and $\mathcal{U}^{-1}(z)$ as in Eqs.~\eqref{eq:l(z)inv}--\eqref{eq:u(z)inv}, corresponds to the inverse of $H_{f_1,\dots,f_N}-z$ interpreted as a continuous operator between $\hfrak_{+1}$ and $\hfrak_{-1}$.

We claim that the restriction of $H_{f_1,\dots,f_N}$ to the domain~\eqref{eq:domain} is a self-adjoint operator in $\hfrak$ whose resolvent is given by the restriction of the operator~\eqref{eq:resolvent} to $\hfrak$. Given $\Phi\in\mathcal{D}(H_{f_1,\dots,f_N})$, using Eq.~\eqref{eq:lu} we have
\begin{equation}
	\left(H_{f_1,\dots,f_N}-z_0\right)\Phi=\mathcal{L}(z_0)\mathcal{G}(z_0)\tilde{\Phi},
\end{equation}
and, since $\tilde{\Phi}\in\mathcal{D}(H_{\rm free})$, $\mathcal{G}(z_0)\tilde{\Phi}\in\hfrak$ whence, by Eq.~\eqref{eq:notice}, $\mathcal{L}(z_0)\mathcal{G}(z_0)\tilde{\Phi}\in\hfrak$ as well; therefore, we have $H_{f_1,\dots,f_N}\Phi\in\hfrak$, proving that the operator is indeed well-defined. To prove self-adjointness, we will show that Eq.~\eqref{eq:extended} for $\Psi\in\hfrak$ admits a unique solution in $\mathcal{D}(H_{f_1,\dots,f_N})$. Since $H_{f_1,\dots,f_N}-z$ is invertible as an operator from $\hfrak_{+1}$ to $\hfrak_{-1}\supset\hfrak$, there exists a unique $\Phi\in\hfrak_{+1}$ solving the equality above, given by $\Phi=\mathcal{U}^{-1}(z)\mathcal{G}^{-1}(z)\mathcal{L}^{-1}(z)\Psi$; but, since $\Psi\in\hfrak$, by Eq.~\eqref{eq:notice} $\mathcal{L}^{-1}(z)\Psi\in\hfrak$ as well; thus, recalling that $\mathcal{G}(z)$, as an operator on $\hfrak$, has domain $\mathcal{D}(H_{\rm free})$, its inverse maps $\hfrak$ in $\mathcal{D}(H_{\rm free})$ and therefore $\mathcal{G}^{-1}(z)\mathcal{L}^{-1}(z)\Psi\in\mathcal{D}(H_{\rm free})$. We have proven that the unique solution can be written as $\Phi=\mathcal{U}^{-1}(z)\tilde{\Phi}$ for some $\tilde{\Phi}\in\mathcal{D}(H_{\rm free})$ and thus (recall Remark~\ref{rem:independent}) belongs to $\mathcal{D}(H_{\rm free})$. Since $\Psi\in\hfrak$ was chosen arbitrarily, we have proven that the range of $H_{f_1,\dots,f_N}-z$, as an operator on $\hfrak$ with domain $\mathcal{D}(H_{f_1,\dots,f_N})$, coincides with the whole Hilbert space $\hfrak$ and is thus self-adjoint, finally proving (i).

To prove (ii) just notice that, if $f_1,\dots,f_N\in\hilb$, the second part of Lemma~\ref{lemma:gammas} implies that $\mathcal{U}^{-1}(z)$ maps $\mathcal{D}(H_{\rm free})$ into itself, whence $\mathcal{D}(H_{f_1,\dots,f_N})=\mathcal{D}(H_{\rm free})$: the two domains thus coincide, whence the corresponding operators coincide.

Let us finally prove (iii). Since $\hilb$ is densely embedded in $\hilb_{-1}$, for all $j=1,\dots,N$ there exists a sequence $\{f^i_j\}_{i\in\mathbb{N}}$ such that $\|f^i_j-f_j\|_{-1}\to0$ as $i\to\infty$, and by Prop.~\ref{prop:af} this implies $a^\dag(f^i_j)\to\adag{f_j}$ (resp. $a(f^i_j)\to\a{f_j}$) in the norm of $\mathcal{B}(\fock,\fock_{-1})$ (resp. $\mathcal{B}(\fock_{+1},\fock)$). As a direct consequence, denoting by $\mathcal{S}^i_\nu(z),\mathcal{G}^i_\nu(z)$ the operators associated with the form factors $f^i_1,\dots,f^i_N$ as per Eqs.~\eqref{eq:gnuz}--\eqref{eq:snuz}, we have $\mathcal{S}^i_\nu(z)\to\mathcal{S}_\nu(z)$ in the norm of $\mathcal{B}(\hfrak^{(\nu)})$. This also implies $[\mathcal{G}^i_\nu(z)]^{-1}\to[\mathcal{G}_\nu(z)]^{-1}$ in the norm of $\mathcal{B}(\hfrak^{(\nu)})$: indeed,
\begin{eqnarray}
	[\mathcal{G}_\nu(z)]^{-1}-[\mathcal{G}^{i}_\nu(z)]^{-1}&=&[\mathcal{G}^{i}_\nu(z)]^{-1}\left[\mathcal{G}^{i}_\nu(z)-\mathcal{G}_\nu(z)\right][\mathcal{G}_\nu(z)]^{-1}\nonumber\\
	&=&[\mathcal{G}^{i}_\nu(z)]^{-1}\left[\mathcal{S}_\nu(z)-\mathcal{S}^i_\nu(z)\right][\mathcal{G}_\nu(z)]^{-1}
\end{eqnarray}
and thus
\begin{eqnarray}
	\left\|[\mathcal{G}_\nu(z)]^{-1}-[\mathcal{G}^{i}_\nu(z)]^{-1}\right\|_{\mathcal{B}(\hfrak^{(\nu)})}&\leq&\left\|[\mathcal{G}^{i}_\nu(z)]^{-1}\right\|_{\mathcal{B}(\hfrak^{(\nu)})}\left\|\left[\mathcal{S}_\nu(z)-\mathcal{S}^i_\nu(z)\right][\mathcal{G}_\nu(z)]^{-1}\right\|_{\mathcal{B}(\hfrak^{(\nu)})}\nonumber\\
	&\leq&\frac{1}{|\!\Im z|}\left\|\left[\mathcal{S}_\nu(z)-\mathcal{S}^i_\nu(z)\right][\mathcal{G}_\nu(z)]^{-1}\right\|_{\mathcal{B}(\hfrak^{(\nu)})}\to0,
\end{eqnarray}
where we also used the bound~\eqref{eq:invgnorm} in Lemma~\ref{lemma:concatenated}. This also readily implies that the corresponding operators $\gamma^i_{\nu,\nu'}(z)$ as in Eqs.~\eqref{eq:gammanu}--\eqref{eq:gammanu4} satisfy $\gamma^i_{\nu,\nu'}(z)\to\gamma_{\nu,\nu'}(z)$, whence, by Eq.~\eqref{eq:resolvent}, norm resolvent convergence is ensured.
\end{proof}

\begin{remark}
	It is instructive to take a closer look to $\mathcal{D}(H_{f_1,\dots,f_N})$. Using the explicit expression~\eqref{eq:u(z)inv}, the most generic vector in the domain can be expressed as
	\small\begin{equation}
	\left[\begin{array}{rrrrrrrrrrr}
		\vphantom{\vdots}\tilde{\Phi}_0&+&\gamma_{0,1}^\dag(z_0)\tilde{\Phi}_1&+&\gamma_{0,2}^\dag(z_0)\tilde{\Phi}_2&+&\ldots&+&\gamma_{0,N-1}^\dag(z_0)\tilde{\Phi}_{N-1}&+&\gamma_{0,N}^\dag(z_0)\tilde{\Phi}_N\\\hline 
		&&\vphantom{\vdots}\tilde{\Phi}_1&+&\gamma_{1,2}^\dag(z_0)\tilde{\Phi}_2&+&\ldots&+&\gamma_{1,N-1}^\dag(z_0)\tilde{\Phi}_{N-1}&+&\gamma_{1,N}^\dag(z_0)\tilde{\Phi}_N \\\hline
		&&&&\vphantom{\vdots}\tilde{\Phi}_2&+&\ldots&+&\gamma_{2,N-1}^\dag(z_0)\tilde{\Phi}_{N-1}&+&\gamma_{2,N}^\dag(z_0)\tilde{\Phi}_N \\\hline
		&&&&&&&&\vdots\quad&&\vdots\quad \\	 \hline
		&&&&&&&&\vphantom{\vdots}\tilde{\Phi}_{N\!-1}&+&\gamma_{N\!-1,N}^\dag(z_0)\tilde{\Phi}_N\\\hline
		&&&&&&&&&&\vphantom{\vdots}\tilde{\Phi}_N
	\end{array}\right]
\end{equation}\normalsize
for some $z_0\in\mathbb{C}\setminus\mathbb{R}$ and $\tilde{\Phi}_\nu\in\mathcal{D}^{(\nu)}$: each component $\Phi_\nu$ of the generic state $\Phi\in\mathcal{D}(H_{f_1,\dots,f_N})$ is given by
\begin{equation}
	\Phi_\nu=\tilde{\Phi}_\nu+\sum_{\nu'=\nu+1}^N\gamma_{\nu,\nu'}^\dag(z_0)\tilde{\Phi}_{\nu'},
\end{equation}
thus effectively involving a coupling between the $\nu$--excited spins component and the ones in which $\nu'>\nu$ spins are excited. Like in the single-spin case, the role of such terms is to ``cancel out'' the ``divergent'' terms (i.e.~those outside $\hfrak$) produced by the creation operators contained in the operators $A_{\nu,\nu+1}^\dag$; those terms are not needed in the regular case.
\end{remark}

\begin{remark}\label{rem:alternative}
An alternative representation of the domain of $H_{f_1,\dots,f_N}$ can be obtained. The operator in Eq.~\eqref{eq:uz} can be more compactly written as $\mathcal{U}(z)=\mathrm{I}-\mathcal{G}^{-1}(z)A^\dag$, with $A$ as in Eq.~\eqref{eq:a}. As such, the domain of $H_{f_1,\dots,f_N}$ can be written as
\begin{equation}
	\mathcal{D}(H_{f_1,\dots,f_N})=\left\{\Phi=\left[\mathrm{I}-\mathcal{G}^{-1}(z_0)A^\dag\right]^{-1}\tilde{\Phi}:\;\tilde{\Phi}\in\mathcal{D}(H_{\rm free})\right\},
\end{equation}
or in the equivalent implicit form
\begin{eqnarray}\label{eq:implicit}
	\mathcal{D}(H_{f_1,\dots,f_N})=\left\{\Phi\in\hfrak_{+1}:\;\left[\mathrm{I}-\mathcal{G}^{-1}(z_0)A^\dag\right]\Phi\in\mathcal{D}(H_{\rm free})\right\},
\end{eqnarray}
where we have used the fact that $\mathcal{U}^{-1}(z_0)$ has values in $\hfrak_{+1}$, whence necessarily $\Phi\in\hfrak_{+1}$. This representation accounts for a direct comparison with the abstract results of~\cite{posilicano2020self} via the general approach briefly recalled in Section~\ref{sec:intro}, as well as those obtained for other models via the IBC method, e.g. in~\cite{lampart2018particle,lampart2019nelson}. While leaving a detailed discussion on this point for future research, we point out that, using the notation of the present paper, the second claim of~\cite[Theorem 3.13]{posilicano2020self} can be stated as follows:\footnote{The operator $T$ in the statement of~\cite[Theorem 3.13]{posilicano2020self} corresponds, in our case, to the bounded operator $A(H_{\rm free}-z_0)^{-1}A^\dag$. Pay attention to the fact that the spaces denoted as $\hfrak_s$ in this paper correspond to those denoted by $\hfrak_{s/2}$ in~\cite{posilicano2020self}.} the restriction of $H_{\rm free}+A+A^\dag$ to the domain
\begin{equation}\label{eq:implicit2}
	\hat{\mathcal{D}}(H_{f_1,\dots,f_N})=\left\{\Phi\in\hfrak_{+1}:\;\left[\mathrm{I}-(H_{\rm free}-z_0)^{-1}A^\dag\right]\Phi\in\mathcal{D}(H_{\rm free})\right\}
\end{equation}
is self-adjoint, provided that both the kernel and the range of $A\restriction\mathcal{D}(H_{\rm free})$ are dense in $\hfrak$. 

In fact, it is easy to see $\mathcal{D}(H_{f_1,\dots,f_N})=\hat{\mathcal{D}}(H_{f_1,\dots,f_N})$. Indeed, the operator $\mathcal{G}(z_0)-(H_{\rm free}-z_0)^{-1}$ is block diagonal, with each block being given by $\mathcal{G}_\nu(z_0)-(H_\nu-z_0)^{-1}$; but, recalling the recursive definition~\eqref{eq:gnuz} and the second resolvent identity,
\begin{equation}
	\mathcal{G}_\nu^{-1}(z_0)-(H_\nu-z_0)^{-1}=-\mathcal{G}_\nu(z_0)^{-1}S_\nu(z_0)(H_\nu-z_0)^{-1};
\end{equation}
since $(H_\nu-z_0)^{-1}:\hfrak_{-1}\rightarrow\hfrak_{+1}$, $\mathcal{S}_\nu(z)$ is bounded in $\hfrak$, and $\mathcal{G}_\nu(z_0)^{-1}$ maps $\hfrak$ in  $\mathcal{D}(H_{\rm free})$, the operator above maps $\hfrak_{-1}$ in $\mathcal{D}(H_{\rm free})$. This implies that, given $\Phi\in\hfrak_{+1}$, the difference between $[\mathrm{I}-(H_{\rm free}-z_0)^{-1}A^\dag]\Phi$ and $[\mathrm{I}-\mathcal{G}^{-1}(z_0)A^\dag]\Phi$ is always in $\mathcal{D}(H_{\rm free})$, thus proving the equivalence between the representations. This generalizes what was observed in~\cite{lonigro2022renormalization} for the case $N=1$.
\end{remark}

\paragraph*{Acknowledgments}
We thank Sascha Lill for fruitful discussions. This work was partially supported by “Istituto Nazionale di Fisica Nucleare” (INFN) through the project “QUANTUM” and the Italian National Group of Mathematical Physics (GNFM-INdAM). We acknowledge financial support by MIUR via PRIN 2017 (Progetto di Ricerca di Interesse Nazionale), project QUSHIP (2017SRNBRK), and by European Union–NextGenerationEU (CN00000013 – “National Centre for HPC, Big Data and Quantum Computing”). 
	
\paragraph{Data availability statement}
Data sharing not applicable to this article as no datasets were generated or analyzed during the current study.

\paragraph{Competing interests} The author has no competing interests to declare that are relevant to the content of this article.
	
\AtNextBibliography{\small}
\DeclareFieldFormat{pages}{#1}\sloppy 
\printbibliography

\end{document}